\documentclass[a4paper,11pt]{article}

\usepackage{microtype}
\usepackage[a4paper,margin=2.5cm]{geometry}
\usepackage[numbers,sort&compress]{natbib}
\usepackage[table]{xcolor}
\usepackage{amsmath}
\usepackage{amssymb}
\usepackage{mathrsfs}
\usepackage{amsthm}
\usepackage{enumitem}
\usepackage{graphicx}
\usepackage[small]{caption}
\usepackage{booktabs}
\usepackage{tikz-cd}
\usetikzlibrary{arrows.meta,backgrounds}
\usepackage{hyperref}
\usepackage{orcidlink}
\usepackage{doi}
\usepackage[textsize=tiny]{todonotes}
\usepackage[framemethod=tikz]{mdframed}
\usepackage{needspace}
\usepackage{mathtools}
\usepackage{xspace}
\usepackage{multirow}
\usepackage{array}
\usepackage{aliascnt}

\theoremstyle{plain}
\newtheorem{theorem}{Theorem}[section]
\newaliascnt{lemma}{theorem}
\newtheorem{lemma}[lemma]{Lemma}
\aliascntresetthe{lemma}
\newaliascnt{corollary}{theorem}
\newtheorem{corollary}[corollary]{Corollary}
\aliascntresetthe{corollary}

\theoremstyle{definition}
\newaliascnt{definition}{theorem}
\newtheorem{definition}[definition]{Definition}
\aliascntresetthe{definition}

\theoremstyle{remark}
\newtheorem*{remark}{Remark}

\usepackage[capitalize,nameinlink]{cleveref}
\crefname{lemma}{Lemma}{Lemmas}
\Crefname{lemma}{Lemma}{Lemmas}
\crefname{corollary}{Corollary}{Corollaries}
\Crefname{corollary}{Corollary}{Corollaries}
\crefname{definition}{Definition}{Definitions}
\Crefname{definition}{Definition}{Definitions}

\newenvironment{mycover}
{\list{}{\listparindent 0pt
        \itemindent    \listparindent
        \leftmargin    1cm
        \rightmargin   1cm
        \parsep        0pt}%
    \raggedright
    \item\relax}
{\endlist}

\newcommand{\myemail}[1]{\,$\cdot$\, {\small #1}}
\newcommand{\myaff}[1]{\,$\cdot$\, {\small #1}\par\smallskip}
\newcommand{\LOCAL}{\ensuremath{\mathsf{LOCAL}}\xspace}
\newcommand{\CONGEST}{\ensuremath{\mathsf{CONGEST}}\xspace}

\definecolor{darkgreen}{rgb}{0,0.5,0}
\definecolor{darkred}{rgb}{0.4,0,0}
\definecolor{darkblue}{rgb}{0,0.1,0.5}

\definecolor{figBlue}{HTML}{24689A}
\definecolor{figTeal}{HTML}{16816A}
\definecolor{figOrange}{HTML}{BB681C}
\definecolor{figPurple}{HTML}{8A639F}
\definecolor{figSky}{HTML}{62A4C5}
\definecolor{figVermilion}{HTML}{C66B58}
\definecolor{figGray}{HTML}{77828C}

\hypersetup{
    colorlinks=true,
    linkcolor=darkred,
    citecolor=darkgreen,
    filecolor=black,
    urlcolor=[rgb]{0,0.1,0.5},
    pdftitle={Distributed (Delta+1)-Coloring in O(sqrt(Delta) + log* n) Rounds},
    pdfauthor={Marc Fuchs, Fabian Kuhn}
}

\begin{document}

\hypersetup{pageanchor=false}
\setcounter{page}{0}
\thispagestyle{empty}

\begin{mycover}
    {\LARGE\bfseries\boldmath
    Distributed $(\Delta+1)$-Coloring in\\
    $O(\sqrt{\Delta}+\log^* n)$ Rounds\par}
    \bigskip
    \bigskip
    \bigskip

    \textbf{Marc Fuchs}\,\orcidlink{0000-0003-2272-4483}
    \myemail{marc.fuchs@cs.uni-freiburg.de}
    \myaff{University of Freiburg}

    \textbf{Fabian Kuhn}\,\orcidlink{0000-0002-1025-5037}
    \myemail{kuhn@cs.uni-freiburg.de}
    \myaff{University of Freiburg}

\end{mycover}

\bigskip

\begin{abstract}
We give a deterministic distributed algorithm for $(\Delta+1)$-coloring
an $n$-node graph of maximum degree $\Delta$ in
$O(\sqrt\Delta+\log^* n)$ rounds in both the \LOCAL and \CONGEST
models. More generally, $(\deg+1)$-list coloring with lists from a
color space of size $C\geq\Delta+1$ can be solved in
$O(\sqrt\Delta\log(1+C/\Delta)+\log^* n)$ rounds; in \CONGEST, we
assume $C\leq n^{O(1)}$.
The algorithm combines Barenboim's polynomial coloring procedure
with arbdefective decompositions, degree reduction, and a short family
of hash functions. Using the polynomial procedure in the
degree-reduction framework already gives an
$O(\sqrt\Delta\log\Delta+\log^* n)$-round algorithm when
$C=O(\Delta)$. Hashing adapts the color space to the current degree
and removes the logarithmic factor. We prove the existence of the
hash family by the probabilistic method and compute it by deterministic
exhaustive search, using unrestricted local computation. We also give
an explicit variant based on lossless expanders that avoids this search
and achieves
$O(\sqrt\Delta+\log^* n)$ rounds for $C=O(\Delta)$.
\end{abstract}

\clearpage
\setcounter{page}{0}
\thispagestyle{empty}

\setcounter{tocdepth}{2}
\hypersetup{bookmarksdepth=3}
\begingroup
\small
\tableofcontents
\endgroup
\clearpage
\hypersetup{pageanchor=true}
\setlength{\parskip}{0.75pt plus 1pt}

\section{Introduction and Related Work}
\label{sec:introduction}

Distributed graph coloring is one of the core problems in the area of
distributed graph algorithms and, more generally, in distributed computing.
In the standard setting, the network is modeled as an $n$-node graph
$G=(V,E)$, and the objective is to compute a proper coloring of the nodes
of $G$. If the maximum degree of $G$ is $\Delta$, the most common task is
to color $G$ with $\Delta+1$ colors, matching the guarantee of a simple
sequential greedy algorithm.

The problem is usually studied in a synchronous message-passing model~\cite{Linial1992,Peleg2000}.
Each node has a unique $O(\log n)$-bit identifier and initially knows only
its own identifier, its incident communication links, its local input, and
any explicitly assumed global parameters. In each round, every node can
send a message to each of its neighbors, receive their messages, and
perform local computations.
At the end, each node must know its own color. The time complexity is the
number of rounds until all nodes terminate. Without a restriction on
message sizes, this is the \LOCAL model; restricting each message to
$O(\log n)$ bits gives the \CONGEST model. In both models, local
computation is unrestricted.

There are two natural ways to study the complexity of distributed
coloring: as a function of the number $n$ of nodes and as a function of
the maximum degree $\Delta$. When considering the complexity as a function
of $\Delta$, the dependence on $n$ is restricted to an additive
$O(\log^* n)$ term. This dependence on $n$ is optimal even for
3-coloring a cycle (i.e., even for $\Delta=2$), as was shown by
Linial~\cite{Linial1992} for deterministic algorithms and by
Naor~\cite{Naor1991} for randomized algorithms. Concentrating on these two
complexity regimes is motivated, in particular, by the gap theorem of
Chang, Kopelowitz, and Pettie~\cite{ChangKopelowitzPettie2019}.
For locally checkable labeling (LCL) problems on bounded-degree graphs in
the \LOCAL model, every deterministic $o(\log n)$-round algorithm can be
sped up to run in $O(\log^* n)$ rounds. Moreover, a randomized
$o(\log\log n)$-round algorithm that succeeds with high probability implies
a deterministic $O(\log^* n)$-round algorithm. Thus, for problems such as
$(\Delta+1)$-coloring, which belong to this fast class for every constant
$\Delta$, the dependence on $\Delta$ is a central part of the complexity
question. In this paper, we focus on the dependence on $\Delta$.

In this context, we also consider the more general
\emph{$(\deg+1)$-list-coloring problem}, where each node $v$ is given a list
$L_v$ of at least $\deg_G(v)+1$ available colors as input and has to choose
its color from $L_v$ so that the resulting coloring is proper. This problem
can also be solved by a simple sequential greedy algorithm. Moreover, it
arises naturally when computing
a $(\Delta+1)$-coloring in stages, where in each stage some uncolored nodes
become colored. When extending a partial coloring, removing the colors of
already colored neighbors leaves every uncolored node $v$ with a list of
still available colors whose size is at least the number of its uncolored
neighbors plus one.

Before discussing the complexity as a function of $\Delta$, we briefly
review what is known as a function of $n$. For a long time, the best known
deterministic and randomized bounds for $(\Delta+1)$-coloring in \LOCAL were
$2^{O(\sqrt{\log n})}$ and $O(\log n)$, respectively. The deterministic
bound was based on the network decomposition algorithm of Panconesi and
Srinivasan~\cite{PanconesiSrinivasan1996}, and the randomized bound follows
from repeated randomized color trials~\cite{Johansson1999} or from a
reduction to Luby's maximal independent set (MIS)
algorithm~\cite{Luby1986}. Improving these bounds was a longstanding open
problem. The first $\operatorname{poly}(\log n)$-round deterministic
algorithm was given by Rozho\v{n} and Ghaffari~\cite{RozhonGhaffari2020} as a
consequence of their breakthrough work on deterministic network
decompositions. The first randomized $o(\log n)$-round algorithm was
obtained by Harris, Schneider, and Su~\cite{HarrisSchneiderSu2016}.

The best known deterministic and randomized upper bounds
for $(\Delta+1)$-coloring in \LOCAL are
$\widetilde O(\log^{5/3} n)$ and
$\widetilde O(\log^{5/3}\log n)$, respectively. Here and throughout,
$\widetilde O(x)$ hides factors at most polylogarithmic in $x$.
The deterministic bound follows from the MIS algorithm of Ghaffari and
Grunau~\cite{GhaffariGrunau2024}. The randomized bound combines this result
with the shattering algorithm of Chang, Li, and
Pettie~\cite{ChangLiPettie2018}. In \CONGEST, the corresponding bounds are
$O(\log^2 \Delta\cdot\log n)$~\cite{GhaffariKuhn2021} and
$O(\log^3\log n)$~\cite{HalldorssonNolinTonoyan2022}, using deterministic
local rounding and communication-efficient randomized coloring,
respectively. The $n$-dependent \LOCAL bounds also extend to the more general
$(\deg+1)$-list-coloring problem. The deterministic bound already holds for
this problem~\cite{GhaffariGrunau2024}, and the randomized bound follows by
combining it with the algorithm of Halld\'orsson, Kuhn, Nolin, and
Tonoyan~\cite{HalldorssonKuhnNolinTonoyan2022}. In \CONGEST, the
deterministic $O(\log^2\Delta\cdot\log n)$ bound extends to lists from a
color space of size at most $\operatorname{poly}(\Delta)$~\cite{GhaffariKuhn2021}.
For a color space of size at most $\operatorname{poly}(n)$, the randomized
algorithm of Halld\'orsson, Nolin, and Tonoyan~\cite{HalldorssonNolinTonoyan2022},
combined with the network decomposition of Ghaffari et al.~\cite{GhaffariGrunauHaeuplerIlchiRozhon2023},
solves $(\deg+1)$-list coloring in $\widetilde O(\log^3\log n)$ rounds. These randomized bounds hold with high probability.

The development of deterministic algorithms with an optimal $O(\log^* n)$ dependence on $n$
started with Linial's $O(\log^* n)$-round $O(\Delta^2)$-coloring
algorithm~\cite{Linial1992} and with the
$O(\Delta^2+\log^* n)$-round $(\Delta+1)$-coloring algorithm of Goldberg,
Plotkin, and Shannon~\cite{GoldbergPlotkinShannon1988}. Together with
repeated applications of a simple one-round color reduction, Linial's algorithm
gives a $(\Delta+1)$-coloring in
$O(\Delta\log\Delta+\log^* n)$
rounds~\cite{SzegedyVishwanathan1993,KuhnWattenhofer2006}.
Barenboim and Elkin, and independently Kuhn, reduced this to
$O(\Delta+\log^* n)$~\cite{BarenboimElkin2009,Kuhn2009,barenboim14distributed}.
These papers introduced defective coloring as a tool for distributed
symmetry breaking. In a $d$-defective coloring, each node has at most $d$ neighbors
with the same color. Many subsequent deterministic distributed coloring algorithms
use some form of defective coloring as a subroutine. The first algorithm
for $(\Delta+1)$-coloring with a sublinear dependence on $\Delta$ was obtained by
Barenboim~\cite{barenboim16sublinear}. He obtained a $(\Delta+1)$-coloring
in time $\widetilde O(\Delta^{3/4})+O(\log^* n)$ and a
$(1+o(1))\Delta$-coloring in time $\widetilde O(\sqrt{\Delta})+O(\log^* n)$.
Barenboim's algorithms also work in the \CONGEST model. Fraigniaud, Heinrich, and
Kosowski~\cite{fraigniaud16local} subsequently extended Barenboim's
framework and improved the bound for $(\Delta+1)$-coloring to
$\widetilde O(\sqrt{\Delta})+O(\log^* n)$ in \LOCAL.
Subsequent improvements to the decomposition and list-coloring steps
led to the bound
\[
    O\bigl(\sqrt{\Delta\log\Delta}+\log^* n\bigr)
\]
for $(\deg+1)$-list coloring in \LOCAL when the color space has size at most
$2^{\operatorname{poly}(\Delta)}$~\cite{fraigniaud16local,BEG18,MausTonoyan20}.

The bandwidth restriction is relevant for these algorithms. The fast
local conflict-coloring routines communicate descriptions of sets that
can be much larger than $O(\log n)$ bits. Fuchs and
Kuhn~\cite{FuchsK23} obtained the first
$\widetilde O(\sqrt{\Delta})+O(\log^* n)$-round \CONGEST algorithm
for $(\Delta+1)$-coloring, based on list defective coloring.
For $(\deg+1)$-list coloring with a color space of size $O(\Delta)$
(and thus in particular for the standard $(\Delta+1)$-coloring problem),
their bound is
\[
    O\bigl(\sqrt{\Delta}\,\log^2\Delta\,\log^6\log\Delta
        +\log^* n\bigr).
\]
Their subsequent work builds on the two-sweep framework of Barenboim and
Elkin~\cite{BarenboimElkin2009} to give a conceptually simpler list defective
coloring algorithm and an alternative
$\widetilde O(\sqrt{\Delta})+O(\log^* n)$-round
\CONGEST algorithm~\cite{FuchsK26}.

\subsection{Our Main Contribution}
One of the main open questions in the context of local distributed graph
algorithms is whether one can improve the deterministic $\Delta$-complexity
of $(\Delta+1)$-coloring to $o(\sqrt{\Delta})$, while retaining an additive
$O(\log^* n)$ term. The above bounds, however, also
leave a more immediate open question: can we remove the remaining
polylogarithmic factors in the $\Delta$-complexity of $(\Delta+1)$-coloring
and obtain $O(\sqrt{\Delta}+\log^* n)$ rounds in both the \LOCAL and
\CONGEST models?

As our main contribution, we answer this question in the affirmative.
More generally, we prove the following result.

\begin{theorem}
\label{thm:main}
The $(\deg+1)$-list-coloring problem on $n$-node graphs of maximum degree
$\Delta$ can be solved deterministically in
\[
    O\left(\sqrt{\Delta}\cdot \log\left(1+\frac{C}{\Delta}\right)
        +\log^* n\right)
\]
rounds in the \LOCAL model, where $C\geq \Delta+1$ is the size of the color space.
The same bound holds in the \CONGEST model if $C\leq\operatorname{poly}(n)$.
For $C=O(\Delta)$ (and thus in particular for the standard $(\Delta+1)$-coloring problem), the round complexity is $O(\sqrt{\Delta}+\log^* n)$.
\end{theorem}

In the following, we describe the main ideas that lead to \cref{thm:main}.
We start by discussing the previous sublinear-in-$\Delta$ algorithms in more detail.

\subsection{Previous Sublinear-in-\texorpdfstring{$\Delta$}{Δ} Algorithms}

The sublinear-in-$\Delta$ algorithms discussed above build on the framework
introduced by Barenboim~\cite{barenboim16sublinear}. As a first step, the
algorithm partitions the graph into node-disjoint subgraphs, each with an
orientation of small maximum outdegree. For this, it uses
\emph{arbdefective colorings}, which were introduced by Barenboim and
Elkin~\cite{BarenboimElkin2011}.
A $\beta$-arbdefective coloring partitions the nodes into color classes,
each inducing a graph of arboricity at most $\beta$. The arbdefective
coloring algorithm of~\cite{barenboim16sublinear} computes such a partition
into $O((\Delta/\beta)\operatorname{polylog}\Delta)$ parts in
$O((\Delta/\beta)\operatorname{polylog}\Delta+\log^* n)$ rounds. It also provides each part
with an orientation of maximum outdegree at most $\beta$.

Starting with an uncolored graph, the algorithm colors the
$O((\Delta/\beta)\operatorname{polylog}\Delta)$ parts sequentially, extending the existing
partial coloring at each step. Coloring each part requires solving a
list-coloring problem: each node must choose a color not already used by
a colored neighbor. The key ingredient of Barenboim's algorithm, and of
the subsequent algorithms, is a fast list-coloring algorithm for graphs
with small outdegree, provided that the nodes have sufficiently large lists.

We describe Barenboim's list-coloring procedure in detail because it is
also the basis of our algorithm. The procedure uses polynomials over a
finite field.\footnote{Barenboim~\cite{barenboim16sublinear} presents the procedure
as extending a partial $C$-coloring, with colors already used by neighbors treated as
forbidden colors. We state it directly for lists of available colors from
$\{1,\ldots,C\}$, which is equivalent.}
Using the outdegree version of Linial's color reduction
(see Section~2.3 of the full version of~\cite{MausTonoyan20}), we compute,
in parallel on the oriented parts, a proper auxiliary $(\eta\beta)^2$-coloring
of each part in $O(\log^* n)$ rounds, for a suitable constant $\eta>0$.
For the list-coloring instance on one of the outdegree-$\beta$ parts,
assume that every list $L_v$ consists of colors from
$\{1,\ldots,C\}$. We choose a finite field
$\mathbb F_q$ with $q\geq\eta\beta$ and $q^2\geq C$, and identify the
$C$ colors with distinct pairs in $\mathbb F_q^2$. Each node $v$ also encodes
its auxiliary color as a pair $(a_v,b_v)\in\mathbb F_q^2$ and considers
the polynomials\footnote{This construction extends to polynomials
of any fixed degree $r$. We later use this generalization
with $r=4$; see \cref{sec:algebraic}.}
$p_v(x)=a_vx^2+b_vx+c_v$ for all $c_v\in \mathbb F_q$.
The $q$ sets $\{(x,p_v(x)):x\in\mathbb F_q\}$ obtained by varying
$c_v$ partition $\mathbb F_q^2$. Node $v$ chooses $c_v$ to maximize
the number of colors from its current list $L_v$ in $\{(x,p_v(x)):x\in\mathbb F_q\}$. This set
therefore contains at least $|L_v|/q$ available colors.
Adjacent nodes have distinct polynomials because their auxiliary colors
differ. The associated sets therefore share at most two elements: the
difference of two distinct polynomials of degree at most two has at most
two roots. Thus, if $|L_v|>2q\beta$, node $v$ can choose an available color
from its set that belongs to none of the sets chosen by its at most
$\beta$ outneighbors.
Exchanging the polynomials and announcing the chosen colors takes
$O(1)$ rounds per part.

For $C=\Theta(\Delta)$ and $\beta=O(\sqrt\Delta)$, we can choose
$q=\Theta(\sqrt\Delta)$, so lists larger than a suitable constant times
$\beta\sqrt\Delta$ suffice. With $(1+\varepsilon)\Delta$ colors, every
node always has at least $\varepsilon\Delta$ available colors. Choosing
$\beta=\Theta(\varepsilon\sqrt\Delta)$ with a sufficiently small constant
therefore gives an algorithm running in $\widetilde O(\sqrt\Delta)+O(\log^* n)$ rounds
for any fixed $\varepsilon>0$. To obtain a $(\Delta+1)$-coloring, Barenboim's algorithm
first computes a proper coloring with
$\Delta+\Theta(\beta\sqrt\Delta)$ colors. It then recolors the nodes
in the extra color classes using the first $\Delta+1$ colors, one
class per round. Balancing the number $\Delta/\beta$ of
parts against the number $\beta\sqrt\Delta$ of extra colors, up to
polylogarithmic factors, gives $\beta=\Theta(\Delta^{1/4})$.
This yields a deterministic $(\Delta+1)$-coloring algorithm with round
complexity $\widetilde O(\Delta^{3/4})+O(\log^* n)$.
Barenboim's algorithm also works in \CONGEST, since the polynomials of
degree at most two can be described with $O(\log q)=O(\log n)$ bits.

Fraigniaud, Heinrich, and Kosowski~\cite{fraigniaud16local} extended the
framework in two ways. They introduced a generalization of ordinary list
coloring called \emph{conflict coloring} and gave a fast algorithm for
solving it in graphs with an outdegree-$\beta$ orientation. For sufficiently
large $\beta$, their list-coloring subroutine requires $O(\log^*\Delta)$
\LOCAL rounds when given a proper auxiliary $O(\Delta^2)$-coloring and lists of size at least
$10\beta^2\ln\Delta$. These bounds are independent of the color-space size
$C$. Their second extension concerns how the parts are processed.
Instead of fully coloring all parts in one pass, they color only nodes
with sufficiently large remaining lists and choose $\beta$ so that each pass halves the maximum
uncolored degree. The algorithm then recurses on the remaining graph. Overall, this results in an
$O(\sqrt\Delta\,\log^{5/2}\Delta+\log^*n)$-round $(\deg+1)$-list-coloring
algorithm. The bound was later improved in two steps. Barenboim, Elkin, and
Goldenberg~\cite{BEG18} improved the arbdefective coloring algorithm to obtain a
decomposition into $O(\Delta/\beta)$ parts, each with an orientation of
maximum outdegree at most $\beta$, computable in $O(\Delta/\beta+\log^*n)$ rounds.
This reduces the overall bound to
$O(\sqrt{\Delta\log\Delta}\,\log^*\Delta+\log^*n)$.
Maus and Tonoyan~\cite{MausTonoyan20} then gave a two-round list-coloring
subroutine for outdegree-$\beta$ graphs with the same auxiliary coloring,
lists of size $\Omega(\beta^2\log\Delta)$ with a sufficiently large constant,
and a color space of size
$C\leq2^{\operatorname{poly}(\Delta)}$.
This removes the $\log^*\Delta$ factor, giving the best previous \LOCAL bound of
$O(\sqrt{\Delta\log\Delta}+\log^*n)$ rounds for $(\deg+1)$-list coloring and
standard $(\Delta+1)$-coloring.

The large messages used by these list-coloring procedures prevent a
direct implementation in \CONGEST. Fuchs and Kuhn~\cite{FuchsK23}
overcame this by combining a list version of defective coloring with
color-space reduction. This allows the bounded-outdegree parts to be
colored in several steps, each using a limited number of colors and thus
smaller messages. They obtained a \CONGEST algorithm running in
$\widetilde O(\sqrt\Delta)+O(\log^*n)$ rounds for $(\deg+1)$-list coloring when
$C\leq\operatorname{poly}(\Delta)$. For $C=O(\Delta)$, the bound is
$O(\sqrt\Delta\,\log^2\Delta\,\log^6\log\Delta+\log^*n)$.

\subsection{Combining Barenboim's Algorithm with Degree Reduction}
We now describe the main ideas that lead to our improved bounds.
Our first observation is that Barenboim's algorithm~\cite{barenboim16sublinear},
described above, can be used directly in the
degree-reduction framework of Fraigniaud, Heinrich, and
Kosowski~\cite{fraigniaud16local}. The framework colors the graph in stages,
each halving a bound on the maximum degree of the uncolored subgraph.
Assume we still need to color a graph of maximum degree at most
$\Delta'\leq\Delta$. Our goal is to extend the partial coloring
so that the maximum degree of the uncolored subgraph becomes at most
$\Delta'/2$. Note that an uncolored node with at least
$\Delta'/2$ uncolored neighbors also has at least
$\Delta'/2$ available colors.

Recall that, given a suitable auxiliary coloring, Barenboim's procedure colors
outdegree-$\beta$ graphs in constant time if the lists have size
$\Omega(\beta\sqrt C)$, with a sufficiently large constant. As long as
$\Delta'$ exceeds a sufficiently large constant multiple of
$\sqrt C$, we can therefore choose
$\beta=\Theta(\Delta'/\sqrt C)$ so that lists of size
$\Delta'/2$ suffice. The decomposition of Barenboim, Elkin, and
Goldenberg~\cite{BEG18} then gives
$O(\Delta'/\beta)=O(\sqrt C)$ parts. We process these parts
sequentially, coloring only nodes with at least $\Delta'/2$
currently available colors. Every node left uncolored then has fewer
than $\Delta'/2$ available colors and hence fewer than
$\Delta'/2$ uncolored neighbors. Thus, a stage halves the degree
bound in $O(\sqrt C)$ rounds.

For $C=O(\Delta)$, repeating these stages $O(\log\Delta)$ times and using a simple
finishing procedure once the remaining degree is $O(\sqrt\Delta)$ gives
an $O(\sqrt\Delta\log\Delta+\log^* n)$-round algorithm. It works in
\CONGEST and already improves on the previous bound in that model.

We can improve this further by stopping after
$O(\log\log\Delta)$ stages, when the maximum uncolored degree is at most
$\Delta/\log^a\Delta$ for a sufficiently large constant $a$.
An existing list-coloring algorithm~\cite{MausTonoyan20,FuchsK23} then
finishes in $O(\sqrt\Delta+\log^* n)$ rounds. Thus, this combination
alone gives an $O(\sqrt\Delta\log\log\Delta+\log^* n)$-round algorithm
in both \LOCAL and \CONGEST, improving on the previous \LOCAL bound
as well.

\subsection{Reducing the Color Space by Hashing}
In the approach above, each degree-halving stage takes
$O(\sqrt C)$ rounds: Barenboim's procedure still uses the original
color space, even as the degree decreases. To make later stages faster,
we use hashing to shrink the color space along with the degree.

For the current degree bound $\Delta'$, set
$k=\Delta'/2$. We use a common sequence of $O(\log(1+C/k))$ hash functions from $[C]$ to
a color space of size $O(k)$ such that every list of at least $k$
colors has more than $k/2$ distinct images under at least one function.
We partition the uncolored graph into $O(\sqrt{\Delta'})$
parts of outdegree $O(\sqrt{\Delta'})$ and process the parts
sequentially. Within each part, we try all hash functions in order,
using Barenboim's procedure to color nodes whose image lists have more
than $k/2$ colors. All active nodes in a call use the same function,
so choosing preimages from their available lists preserves properness.

If a node remained uncolored with at least $k$ available colors,
some function would preserve more than $k/2$ images of its final list.
Since lists only shrink, that function would have allowed us to color
the node when it was tried in the node's part. Thus, every remaining list has fewer
than $k$ colors, and hence every uncolored node has fewer than
$k=\Delta'/2$ uncolored neighbors. Trying all functions in one part
takes $O(\log(1+C/\Delta'))$ rounds. Processing all
$O(\sqrt{\Delta'})$ parts, including the partition computation, takes
$O(\sqrt{\Delta'}\log(1+C/\Delta')+\log^*\Delta)$
rounds per stage. Summing over the degree scales gives \cref{thm:main}.

We prove the existence of the hash family by the probabilistic method.
All nodes can find the same family by deterministic exhaustive search,
since local computation is unrestricted. In
Appendix~\ref{sec:explicit-hashing}, we give an alternative based on explicit
lossless expanders that avoids this search. It replaces the factor
$\log(1+C/\Delta)$ in the round complexity by
$2^{O((\log\log(4+C/\Delta))^3)}$, while retaining
$O(\sqrt\Delta+\log^* n)$ rounds for $C=O(\Delta)$.

\section{Model and Preliminaries}
\label{sec:preliminaries}

\subsection{Basic Definitions}

We use the synchronous \LOCAL and \CONGEST models described in
\cref{sec:introduction}. All graphs are finite, simple, and undirected. As part of our algorithm, we compute edge orientations, which are used to assign responsibility for avoiding color conflicts. When we work with oriented
graphs, communication remains possible in both directions on every edge. Orientations need not be acyclic. We write $N_H(v)$ for the neighbors of $v$ in a graph $H$, and $N_H^+(v)$ for its outneighbors in an oriented graph. We omit the subscript when the graph is clear.

For an integer $a\geq1$, we write $[a]=\{1,\ldots,a\}$. In our list-coloring instances, the color
lists are subsets of $[C]$, and we assume that all nodes know $\Delta$ and $C$.
In \CONGEST, we assume $C\leq n^{O(1)}$. We assume $\Delta\geq2$,
since the case $\Delta\leq1$ can be solved in one round. All logarithms
are binary. The iterated logarithm $\log^* x$ is the minimum number of
successive applications of the logarithm needed to reduce $x$ to at most
one. In asymptotic bounds, we take all logarithms and iterated logarithms
to be at least one.

\subsection{Computing Auxiliary and Arbdefective Colorings}
We distinguish the colors chosen from the input lists from auxiliary
colors used for symmetry breaking and scheduling. The following
standard form of Linial's color reduction will be useful.

\begin{theorem}[Color reduction~\cite{Linial1992}]
\label{thm:linial}
There is an absolute constant $\kappa$ such that a graph of maximum
degree at most $\Delta\geq1$, equipped with a proper $\gamma$-coloring, can be
properly colored with at most $\kappa\Delta^2$ colors in
$O(\log^* \gamma)$ rounds. The algorithm works in \CONGEST when
$\gamma\leq n^{O(1)}$. In particular, unique identifiers give a proper
$O(\Delta^2)$-coloring in $O(\log^* n)$ rounds.
\end{theorem}

A $\beta$-arbdefective $p$-coloring is a partition into $p$ color
classes such that each class induces a graph of arboricity at most
$\beta$, where arboricity is the minimum number of forests needed to
cover the edges. We use the \emph{oriented} version: together with the
partition, the algorithm supplies an orientation of each induced
subgraph with maximum outdegree at most $\beta$. Such an orientation
implies arboricity at most $\beta+1$; conversely, a graph of arboricity
at most $\beta$ has an orientation of outdegree at most $\beta$. These
constant differences are immaterial for our applications. We state
the required decomposition directly in terms of outdegree.

\begin{theorem}[Arbdefective decomposition~\cite{BEG18,MausTonoyan20}]
\label{thm:partition}
Given a graph of maximum degree at most $\Delta$, a proper $\gamma$-coloring,
and an integer $1\leq\beta\leq\Delta$, one can compute a partition
\[
    V=V_1\mathbin{\dot\cup}\cdots\mathbin{\dot\cup}V_p,
    \qquad p=O(\Delta/\beta),
\]
and an orientation of each $G[V_i]$ of maximum outdegree at most
$\beta$, in $O(\Delta/\beta+\log^* \gamma)$ deterministic \CONGEST rounds.
\end{theorem}

For $\beta$ above an absolute constant, this is the formulation of the
result of Barenboim, Elkin, and Goldenberg given in Lemma~4.1 of the
full version of Maus and Tonoyan~\cite{MausTonoyan20}. For smaller
positive $\beta$, compute a proper $O(\Delta)$-coloring in
$O(\Delta+\log^* \gamma)$ \CONGEST rounds using \cref{thm:linial} followed by
the color reduction of~\cite{BEG18}. Its independent color classes
have outdegree zero, and $O(\Delta)=O(\Delta/\beta)$ in this case.
The stated bound also covers one round to inform both endpoints of the
part memberships and edge orientations. Part indices are taken from a
commonly known range of size $O(\Delta/\beta)$, so processing the parts
does not require finding the largest index actually used. Restricting a part to a subset of its
nodes preserves the outdegree bound.

\subsection{Polynomials over Finite Fields}
For every real $x\geq2$, a power of two $q$ can be chosen with
$x\leq q<2x$, and a field $\mathbb F_q$ exists. We use a fixed
representation of this field known to all nodes. For a polynomial
$P\in\mathbb F_q[X]$, let
\begin{equation}
\label{eq:polynomial-graph}
    \Gamma_P:=\{(x,P(x)):x\in\mathbb F_q\}
\end{equation}
be its graph. The following lemma collects the properties used in
the algebraic list-coloring primitive.

\begin{lemma}[Polynomial partitions]
\label{lem:polynomial-partitions}
Let $q$ be a prime power and $r\geq1$ an integer. There are $q^r$
polynomials over $\mathbb F_q$ with zero constant coefficient and
degree at most $r$. For each such polynomial $P$, the $q$ sets
$\Gamma_{P+t}$, $t\in\mathbb F_q$, partition $\mathbb F_q^2$ into
sets of size $q$. Moreover, for distinct such polynomials $P,Q$
and any $s,t\in\mathbb F_q$,
\[
    |\Gamma_{P+s}\cap\Gamma_{Q+t}|\leq r.
\]
\end{lemma}

\begin{proof}
The $r$ nonconstant coefficients (i.e., all coefficients except the additive constant $t$) can be chosen independently,
giving $q^r$ polynomials. Each graph contains one point for each
$x\in\mathbb F_q$, and a point $(x,y)$ belongs to the unique
translate with $t=y-P(x)$. Finally, $P-Q+s-t$ is nonzero because
$P$ and $Q$ differ in a nonconstant coefficient. It has degree at
most $r$ and therefore at most $r$ roots, which bound the number
of intersection points.
\end{proof}

\section{The New Coloring Algorithm}
\label{sec:algorithm}

We first formulate Barenboim's algebraic coloring procedure~\cite{barenboim16sublinear} as a
general list-coloring primitive and combine it with the degree-halving
strategy of~\cite{fraigniaud16local}. We then show how hashing can be used to improve the complexity of the algorithm and prove
\cref{thm:main}. The only external distributed subroutines needed for
the proof of the main theorem are \cref{thm:linial,thm:partition}.

\subsection{Basic Algebraic List Coloring with Bounded Outdegree}
\label{sec:algebraic}

The following primitive is Barenboim's polynomial coloring
procedure~\cite{barenboim16sublinear}, stated for arbitrary color
lists and with an explicit dependence on the color-space size.

We are given an oriented graph $H$, a proper auxiliary $\gamma$-coloring
$\psi$, and a list $L_v\subseteq[C]$ at each node $v$.
Choose a prime power $q$ and an integer $r\geq1$ such that
\[
    C\leq q^2\qquad\text{and}\qquad \gamma\leq q^r.
\]
Represent the colors in $[C]$ by distinct pairs in $\mathbb F_q^2$
and identify colors with these pairs. To each auxiliary color
$i\in[\gamma]$, assign a distinct polynomial $P_i$ with zero constant
coefficient and degree at most $r$. Such an assignment exists by
\cref{lem:polynomial-partitions}. All nodes use the same representation
of colors and assignment of polynomials.

Each node chooses the constant term of its polynomial so that the
corresponding set contains as many available colors as possible.
It then chooses an available color outside the sets selected by its
outneighbors. \Cref{fig:algebraic-choice} illustrates these two steps.
The following pseudocode gives the details.

\begin{mdframed}[linewidth=0.4pt,nobreak=true]
\noindent\textbf{Algebraic list coloring at node $v$}
\begin{enumerate}[label=\arabic*.,leftmargin=*,itemsep=2pt]
    \item \textbf{Choose a set.} Using the graph notation from
    \eqref{eq:polynomial-graph}, choose $t_v\in\mathbb F_q$ maximizing
    $|L_v\cap\Gamma_{P_{\psi(v)}+t_v}|$, and set
    \[
        Q_v=\Gamma_{P_{\psi(v)}+t_v},\qquad S_v=L_v\cap Q_v.
    \]
    \item \textbf{Exchange descriptions.} Send $(\psi(v),t_v)$ to all
    neighbors and reconstruct $Q_u$ for every outneighbor $u$ from
    its message.
    \item \textbf{Choose a color.} Select a color in
    \[
        S_v\setminus\bigcup_{u\in N_H^+(v)}Q_u
    \]
    and output its original label in $[C]$.
\end{enumerate}
\end{mdframed}

\begin{figure}[tbp]
    \centering
    % Exact example: P_v=X^2 over F_7; the seven translate counts are
% (5,4,4,4,4,4,4). The two outneighbors intersect Q_v at x=0,1
% and x=2,3, respectively.
\begin{tikzpicture}[
    x=1cm,y=1cm,font=\footnotesize,
    every node/.style={inner sep=1pt},
    panel title/.style={font=\small\bfseries,anchor=west},
    candidate/.style={circle,draw=figBlue,fill=figBlue!14,
        minimum size=4.4mm,inner sep=0pt,line width=.7pt}]
    \path[use as bounding box] (-.10,-.12) rectangle (15.75,5.05);
    \node[panel title] at (0,4.77) {(a) Choose a set with many available colors};
    \node[panel title] at (8.05,4.77) {(b) Remove conflicts with outneighbors};
    \draw[black!12] (7.55,.15) -- (7.55,4.3);

    % Partition of F_7^2: one background color for each translate.
    \foreach \t/\shade in {0/figBlue,1/figOrange,2/figPurple,
            3/figTeal,4/figSky,5/figVermilion,6/figGray} {
        \foreach \x in {0,...,6} {
            \pgfmathtruncatemacro{\y}{mod(\x*\x+\t,7)}
            \ifnum\t=0
                \path[fill=\shade!17,draw=\shade,line width=.55pt]
                    ({.55+.43*\x-.205},{1.20+.43*\y-.205})
                    rectangle ++(.41,.41);
                \pgfmathtruncatemacro{\available}{\x<5}
            \else
                \path[fill=\shade!18,draw=white,line width=.55pt]
                    ({.55+.43*\x-.205},{1.20+.43*\y-.205})
                    rectangle ++(.41,.41);
                \pgfmathtruncatemacro{\available}{mod(\x-\t+7,7)<4}
            \fi
            \ifnum\available=1
                \fill[black!85] ({.55+.43*\x},{1.20+.43*\y})
                    circle[radius=1.65pt];
            \fi
        }
    }
    \foreach \x in {0,...,6} {
        \node[text=black!65,font=\scriptsize] at ({.55+.43*\x},.81) {$\x$};
        \node[text=black!65,font=\scriptsize] at (.14,{1.20+.43*\x}) {$\x$};
    }
    \node at (3.53,.81) {$x$};
    \node at (.14,4.15) {$y$};
    \node at (1.85,4.20) {$\mathbb F_7^2$};
    \fill[black!85] (.75,.34) circle[radius=1.65pt];
    \node[anchor=west] at (.94,.34) {available color in $L_v$};

    % The counts explain the maximization without a second large grid.
    \node at (4.65,4.20) {shift $t$};
    \node at (6.18,4.20) {available};
    \path[fill=figBlue!7,rounded corners=2pt]
        (3.95,3.55) rectangle (6.95,3.99);
    \foreach \t/\shade in {0/figBlue,1/figOrange,2/figPurple,
            3/figTeal,4/figSky,5/figVermilion,6/figGray} {
        \path[draw=\shade,fill=\shade!18,line width=.55pt]
            (4.15,{3.77-.43*\t}) ++(-.09,-.09) rectangle ++(.18,.18);
        \node at (4.65,{3.77-.43*\t}) {$\t$};
        \ifnum\t=0
            \node[text=figBlue] at (5.27,3.77) {$Q_v$};
            \node[text=figBlue,font=\footnotesize\bfseries] at (6.18,3.77) {5};
        \else
            \node at (6.18,{3.77-.43*\t}) {4};
        \fi
    }
    \node at (5.43,.34) {$|S_v|\geq |L_v|/q$};

    % Five candidates, with exclusions aligned by candidate.
    \foreach \i/\point in {0/{(0,0)},1/{(1,1)},2/{(2,4)},
                           3/{(3,2)},4/{(4,2)}} {
        \node[font=\scriptsize] at ({9.72+1.17*\i},4.20) {$\point$};
        \node[candidate] at ({9.72+1.17*\i},3.65) {};
        \fill[figBlue] ({9.72+1.17*\i},3.65) circle[radius=1.65pt];
    }
    \node[anchor=west] at (8.05,3.65) {$S_v$};
    \node[anchor=west,text=figOrange] at (8.05,2.81) {$Q_{u_1}$};
    \node[anchor=west,text=figPurple] at (8.05,1.97) {$Q_{u_2}$};
    \node[anchor=west,text=figTeal] at (8.05,1.13) {choose};
    \foreach \i in {0,1} {
        \draw[figOrange,line width=1.3pt,line cap=round]
            ({9.72+1.17*\i-.075},2.735) -- ++(.15,.15)
            ({9.72+1.17*\i-.075},2.885) -- ++(.15,-.15);
    }
    \foreach \i in {2,3} {
        \draw[figPurple,line width=1.3pt,line cap=round]
            ({9.72+1.17*\i-.075},1.895) -- ++(.15,.15)
            ({9.72+1.17*\i-.075},2.045) -- ++(.15,-.15);
    }
    \foreach \i in {2,3,4} {
        \node[text=black!35] at ({9.72+1.17*\i},2.81) {--};
    }
    \foreach \i in {0,1,4} {
        \node[text=black!35] at ({9.72+1.17*\i},1.97) {--};
    }
    \node[candidate,draw=figTeal,fill=figTeal!14,line width=1pt]
        at (14.40,1.13) {};
    \fill[figTeal] (14.40,1.13) circle[radius=2.4pt];
    \node at (11.45,.34) {$5-2-2=1$ available color};
\end{tikzpicture}
    \caption{Algebraic list coloring with $q=7$, $r=2$, and two
    outneighbors. (a) The shaded cells show the seven sets
    $\Gamma_{X^2+t}$ defined in \eqref{eq:polynomial-graph}; black dots
    mark the 29 colors in $L_v$.
    The selected set $Q_v=\Gamma_{X^2}$ contains five available colors.
    (b) The outneighbors' sets
    $Q_{u_1}=\Gamma_{2X^2-X}$ and
    $Q_{u_2}=\Gamma_{2X^2+2X+6}$ each exclude two candidates,
    leaving $(4,2)$ available. All arithmetic is over $\mathbb F_7$.}
    \label{fig:algebraic-choice}
\end{figure}
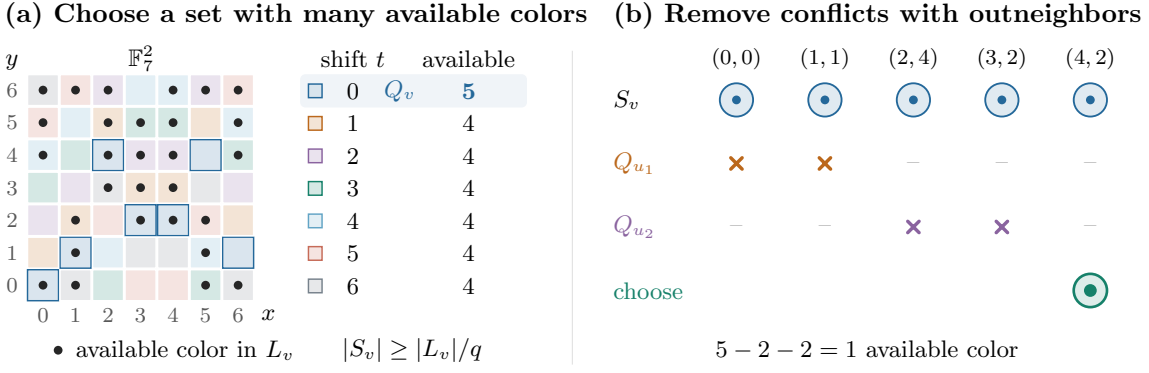

\begin{lemma}[Algebraic list coloring]
\label{lem:algebraic}
If $|L_v|>rq\,|N_H^+(v)|$ for every node $v$ of $H$, the above
algorithm computes a proper list coloring of $H$
in one communication round using $O(\log \gamma+\log q)$-bit messages.
Announcing the selected colors takes one additional round with
$O(\log C)$-bit messages.
\end{lemma}

\begin{proof}
For every node $v$, the partition property of
\cref{lem:polynomial-partitions} gives
\[
    |S_v|\geq |L_v|/q>r|N_H^+(v)|.
\]
Since $\psi$ is proper, the same lemma gives
$|Q_v\cap Q_u|\leq r$ for every edge $uv$ of $H$.
The outneighbors therefore exclude at most
$r|N_H^+(v)|$ elements of $S_v$, leaving an available color in step 3.

For an edge oriented from $v$ to $u$,
the color selected by $v$ lies outside $Q_u$, whereas the color
selected by $u$ lies in $Q_u$. The output is therefore a proper
coloring of $H$ from the original lists.

Only the second step uses communication. Its message consists of an
auxiliary color and a field element and has
$O(\log \gamma+\log q)$ bits. Announcing an actual color costs
$O(\log C)$ bits. No list or candidate set is transmitted.
\end{proof}

For a proper auxiliary coloring with $\gamma=O(C^2)$ colors, a field of
size $q=\Theta(\sqrt C)$ and polynomials of degree $r=4$ suffice.
Thus lists of size $\Omega(\beta\sqrt C)$, with a sufficiently large
constant, suffice for outdegree at most $\beta$. If the auxiliary
coloring already has at most $q^2$ colors, quadratic polynomials suffice, and we can take
$P_i(X)=a_iX^2+b_iX$.

\begin{remark}
In each degree-halving stage below, we could compute a fresh proper
$O(\beta^2)$ auxiliary coloring of each outdegree-$\beta$ part using the
outdegree version of Linial's reduction
(see Section~2.3 of the full version of~\cite{MausTonoyan20}).
Starting from the retained coloring, this takes $O(\log^*\Delta)$ rounds
in parallel over all parts. Since $\beta=O(q)$ in our applications,
increasing $q$ by a constant factor if necessary would then permit
degree-2 polynomials. With hashing, the new auxiliary coloring could be
reused for every hash function in the stage.
This would, however, not improve the asymptotic
round complexity: the smaller degree changes only constant factors in
the list-size requirement, and the additional color reductions fit
within the stated round bounds.
\end{remark}

\subsection{Degree Reduction without Hashing}
\label{sec:basic-degree}

We combine the algebraic procedure with the degree-reduction strategy
of Fraigniaud, Heinrich, and Kosowski~\cite{fraigniaud16local} and the
arbdefective decomposition of Barenboim, Elkin, and
Goldenberg~\cite{BEG18}. In each stage, we color nodes with sufficiently
large remaining lists so that the maximum uncolored degree is halved.

We are given a $(\deg+1)$-list-coloring instance with lists
$L_v\subseteq[C]$, where $C=O(\Delta)$. We maintain a proper partial
coloring and write $L'_v$ for the residual list of an uncolored node
$v$, obtained by removing the colors used by colored neighbors from
$L_v$. Whenever a node is colored, it announces its color to all
neighbors, which update their lists before the next call. These lists
only shrink and always contain more colors than the node has uncolored
neighbors, since each colored neighbor removes at most one color.
We use this convention throughout the remainder of the section.

Initially, every node is uncolored. Compute and retain a proper
auxiliary coloring with $\gamma=O(\Delta^2)$ colors by \cref{thm:linial}.
Fix $r=4$, let $q$ be the least power of two satisfying
\[
    C\leq q^2,\qquad \gamma\leq q^4,
\]
and initialize the degree bound
$\Delta'$ to the least power of two at least $\Delta$.

At degree bound $\Delta'$, we choose the outdegree bound $\beta$
so that the algebraic procedure can color nodes with at least
$\Delta'/2$ available colors. We process the parts sequentially,
always using the current residual lists. Nodes left uncolored then
have fewer than $\Delta'/2$ available colors, so their uncolored
degrees are also smaller than $\Delta'/2$.
We repeat the following stage while
$\Delta'\geq32q$.\footnote{In each algebraic call here and below,
activity and the polynomial translate are chosen locally from the node's
list. In round 1, active nodes send an activity flag and their polynomial
description; inactive nodes send an inactive flag. Active nodes avoid sets
only for active outneighbors in their part. In round 2, newly colored
nodes announce their actual colors to all neighbors.}

\begin{figure}[!t]
\begin{mdframed}[linewidth=0.4pt,nobreak=true]
\noindent\textbf{One degree-halving stage}
\begin{enumerate}[label=\arabic*.,leftmargin=*,itemsep=2pt]
    \item \textbf{Partition the uncolored graph.} Set
    \[
        \beta=\left\lfloor\frac{\Delta'}{16q}\right\rfloor
    \]
    and apply \cref{thm:partition} to the uncolored graph to obtain
    parts $V_1,\ldots,V_p$, each with an orientation of outdegree
    at most $\beta$.
    \item \textbf{Color nodes with large lists.} Process the parts in
    order. In part $V_i$, activate the uncolored nodes with
    $|L'_v|\geq\Delta'/2$ and apply \cref{lem:algebraic} to
    the graph induced by these nodes, with their residual lists
    and the retained auxiliary coloring.
    Announce the selected colors and update the residual lists
    before processing the next part.
    \item \textbf{Update the degree bound.} Replace $\Delta'$
    by $\Delta'/2$.
\end{enumerate}
\end{mdframed}
\end{figure}

Once the loop ends, set $d=\Delta'$. We compute an auxiliary
$2d$-coloring of the remaining graph in $O(\sqrt d+\log^* n)$
\CONGEST rounds, using Theorem~6.4 of the full version
of~\cite{BEG18}. This algorithm colors the parts of an arbdefective
partition using Barenboim's algebraic procedure from
\cref{lem:algebraic}, with parameters chosen for degree bound $d$ and temporary lists
$[2d]$. Every node retains at least $d$ available auxiliary colors,
so one pass through the parts suffices. We then process the $2d$
auxiliary color classes in order, choosing and announcing actual
colors from the residual lists $L'_v$.

\begin{theorem}[The basic algorithm]
\label{thm:basic}
For $C=O(\Delta)$, the above algorithm solves $(\deg+1)$-list coloring
deterministically in $O(\sqrt{\Delta}\log\Delta+\log^* n)$ rounds
in \CONGEST.
\end{theorem}

\begin{proof}
We first analyze a degree-halving stage. The field choice gives
$q=\Theta(\sqrt C)$ because $\gamma=O(\Delta^2)$ and
$C\geq\Delta+1$. The loop condition ensures
$\beta\geq1$ and $\beta=\Theta(\Delta'/q)$, so
\cref{thm:partition} applies and produces $p=O(q)$ parts.
For each active node, the hypothesis of \cref{lem:algebraic} holds:
\[
    rq\beta=4q\beta\leq\Delta'/4
        <\Delta'/2\leq |L'_v|.
\]
Each call therefore gives a proper coloring from the residual lists.
Since these exclude colors already used by colored neighbors, the
partial coloring remains proper.
Every node left uncolored had fewer than $\Delta'/2$ available colors
when its part was processed. Residual lists only shrink and remain
larger than the uncolored degrees, so the maximum uncolored degree
is less than $\Delta'/2$. This justifies the updated degree bound.

The decomposition and the two rounds per part take
\begin{equation}
\label{eq:basic-stage-time}
    O(q+\log^*\Delta)=O(\sqrt C)
\end{equation}
rounds per stage. There are $O(\log\Delta)$ stages, and when the
loop ends, the remaining graph has degree at most $d=O(q)$.

The auxiliary color classes for finishing are independent, and
every residual list remains nonempty. Processing these classes
gives a proper list coloring in at most $2d$ further rounds.
Computing the auxiliary coloring and processing its classes takes
$O(\sqrt d+\log^* n+d)=O(\sqrt C+\log^* n)$ rounds.
Together with the stages and the initial color reduction, the total is
$O(\sqrt\Delta\log\Delta+\log^* n)$ because $C=O(\Delta)$.
All messages contain auxiliary labels, field elements, or actual
colors, each representable in $O(\log n)$ bits.
\end{proof}

\paragraph{Stopping early.}
We can improve the running time by stopping the degree reduction
earlier and using an existing list-coloring algorithm on the remaining
graph. Fix a constant $a>0$ and set
$d=\lceil\Delta/\log^a\Delta\rceil$. Run the basic algorithm,
stopping its degree-halving loop as soon as $\Delta'\leq d$.
If the ordinary stopping condition is reached first, use the
finishing procedure above instead. To finish, choose a
sufficiently large fixed constant $a$ and apply the list-coloring
algorithm of Fuchs and Kuhn~\cite{FuchsK23} to the residual instance.
In \LOCAL, the algorithm of Maus and Tonoyan~\cite{MausTonoyan20}
can also finish the coloring.

\begin{corollary}[Combining with known list-coloring algorithms]
\label{cor:hybrid}
For $C=O(\Delta)$ and any fixed constant $a>0$, the above variant
computes a proper partial list coloring leaving maximum uncolored
degree at most $\lceil\Delta/\log^a\Delta\rceil$ in
$O(\sqrt{\Delta}\log\log\Delta+\log^* n)$ deterministic
\CONGEST rounds. For sufficiently large $a$, completing the coloring
as described above gives a full $(\deg+1)$-list coloring within the
same bound in both \LOCAL and \CONGEST.
\end{corollary}

\begin{proof}
For sufficiently large $\Delta$, the target degree is larger than
the $\Theta(\sqrt\Delta)$ threshold at which the basic loop stops.
Reaching it takes only $O(a\log\log\Delta)$ stages. The first
assertion follows from \eqref{eq:basic-stage-time}. Bounded values of $\Delta$ are covered by the finishing procedure above.

The residual instance is still a $(\deg+1)$-list-coloring instance,
since each colored neighbor removes at most one color. Its degree bound is $d$, and
its color space has size $C=O(\Delta)=d^{O(1)}$. The algorithm of
Fuchs and Kuhn~\cite{FuchsK23} takes
$\sqrt d\,\operatorname{polylog}d+O(\log^* n)$ \CONGEST rounds.
For a sufficiently large fixed $a$, this is
$O(\sqrt\Delta+\log^* n)$, proving the second assertion.
\end{proof}

\subsection{A Short Family of Hash Functions}
\label{sec:hashing}

In the basic degree-reduction scheme, applying the list-coloring
procedure of \cref{lem:algebraic} over a color space of size $C$ costs
$O(\sqrt C)$ rounds per stage. The field size remains
$q=\Theta(\sqrt C)$ even when $C\gg\Delta$, so this cost does not
decrease as the maximum uncolored degree $\Delta'$ decreases.
We use hashing to reduce the color space along with the degree.
For the general argument, we introduce an
\emph{image-preserving hash family}.

\begin{definition}[Image-preserving hash family]
\label{def:image-preserving-hashes}
Let $N,M,R,k$ be positive integers with $k\leq N$. A family of
hash functions $h_1,\ldots,h_R:[N]\to[M]$ is
\emph{image-preserving at threshold $k$} if every set
$S\subseteq[N]$ of size at least $k$ satisfies
\begin{equation}
\label{eq:hash-property}
    \max_{j\in[R]}|h_j(S)|>k/2.
\end{equation}
\end{definition}

We apply such families with $N=C$ and choose the threshold $k$
proportional to the current degree bound $\Delta'$. The construction in
\cref{lem:hashes} gives $M=8k=\Theta(\Delta')$, so $M=o(C)$ whenever
$\Delta'=o(C)$. Choosing $k=\Delta'/2$ allows each stage to halve the
maximum uncolored degree, as shown in \cref{lem:hash-stage}.
In each call, we color nodes whose image lists are sufficiently large
and choose preimages from their current residual lists. All participating
nodes use the same hash function, so adjacent nodes assigned different
image colors choose different original colors.

\begin{lemma}[Image-preserving hash family]
\label{lem:hashes}
For integers $N\geq k\geq2$ with $k$ even, there is an
image-preserving family at threshold $k$
(\cref{def:image-preserving-hashes}) consisting of
\[
    R=\lceil\log(N/k)\rceil+3
    \quad\text{hash functions}\quad h_1,\ldots,h_R:[N]\longrightarrow[8k].
\]
The family can be computed by a finite deterministic procedure using
only $N$ and $k$.
\end{lemma}

\paragraph{Construction.}
Enumerate all $R$-tuples of functions from $[N]$ to $[8k]$ in a fixed
order. Test each tuple on all $k$-element subsets of $[N]$ and return
the first tuple satisfying \eqref{eq:hash-property} on all of them.

\begin{proof}
Choose a random function $h:[N]\to[8k]$ by assigning independent uniform
images to all elements. Fix a $k$-element set $S$. If $|h(S)|\leq k/2$,
all images of $S$ belong to some $(k/2)$-element subset of $[8k]$.
Hence
\[
    \Pr\bigl[|h(S)|\leq k/2\bigr]
    \leq \binom{8k}{k/2}\left(\frac1{16}\right)^k
    \leq (16e)^{k/2}16^{-k}
    =\left(\frac e{16}\right)^{k/2}
    <2^{-k}.
\]
For $R$ independently chosen functions, the probability that every function
fails for this fixed set is less than $2^{-kR}$. A union bound over
all $k$-element sets gives
\[
    \Pr[\text{some $k$-element set fails for every function}]
    <\binom{N}{k}2^{-kR}
    \leq\left(\frac{eN}{k}\,2^{-R}\right)^k
    \leq(e/8)^k<1.
\]
Thus a suitable family exists for all $k$-element sets, and every
larger set contains such a set. The deterministic search therefore
terminates and returns a family satisfying \eqref{eq:hash-property}.
Since the enumeration and all tests depend only on $N$ and $k$,
every node obtains the same family.
\end{proof}

\subsection{Coloring with Compressed Color Spaces}
\label{sec:compressed}

We use image-preserving families to adapt the color space to the
current degree for general $C\geq\Delta+1$. As before, we retain a proper auxiliary
$O(\Delta^2)$-coloring and maintain residual lists $L'_v$ that form
a $(\deg+1)$-list-coloring instance on the uncolored graph.

Consider a stage in which the uncolored graph has maximum degree at
most $\Delta'$, where $\Delta'\geq2$ is a
power of two and $\Delta'/2\leq C$. Set $k=\Delta'/2$
and use a common image-preserving family $h_1,\ldots,h_R:[C]\to[M]$
at threshold $k$, computed locally from $C$ and $k$.
We also reduce the auxiliary color space. Apply
\cref{thm:linial} to the retained coloring on the uncolored graph to
obtain a proper coloring $\psi$ with at most $\kappa(\Delta')^2$
colors. Fix $r=4$, let $q\geq2$ be the least power of two satisfying
\[
    M\leq q^2,\qquad \kappa\cdot(\Delta')^2\leq q^4,
\]
and set $\beta=\lfloor\Delta'/(16q)\rfloor$. Assume $\beta\geq1$ and
apply
\cref{thm:partition} to obtain parts $V_1,\ldots,V_p$ with
outdegree-$\beta$ orientations. Keep these parts and orientations
fixed throughout the stage.

We process the parts sequentially. Within each part, we try all hash
functions in order, coloring nodes whose current lists have
sufficiently many distinct images. After each call, nodes announce
their actual colors to update the remaining lists before the next
hash function or part is processed.

\begin{mdframed}[linewidth=0.4pt,nobreak=true]
\noindent\textbf{One degree-halving stage with hashing}

Iterate over the parts $V_1,\ldots,V_p$ and for each $i=1,\ldots,p$,
try $h_1,\ldots,h_R$ in order on $V_i$
before proceeding to the next part.
For each pair $(i,j)\in [p]\times[R]$, execute the following:
\begin{enumerate}[label=\arabic*.,leftmargin=*,itemsep=2pt]
    \item \textbf{Select active nodes locally.} Activate the uncolored
    nodes $v\in V_i$ with $|h_j(L'_v)|>k/2$.
    \item \textbf{Color the image lists (round 1).} Apply \cref{lem:algebraic}
    to the graph induced by the active nodes, with lists
    $h_j(L'_v)\subseteq[M]$, using $\psi$ and the orientation of
    the part. Each active node $v$ obtains an image color $z_v$
    and chooses an actual color
    $c_v\in L'_v$ with $h_j(c_v)=z_v$.
    \item \textbf{Update the residual lists (round 2).} Permanently color each
    active node $v$ with $c_v$ and announce $c_v$ to all neighbors.
    Uncolored neighbors remove these actual colors from their lists
    before the next call.
\end{enumerate}
\end{mdframed}

\begin{lemma}[Degree halving with hashing]
\label{lem:hash-stage}
Let $h_1,\ldots,h_R:[C]\to[M]$ be an image-preserving hash family
at threshold $k=\Delta'/2$ in the sense of
\cref{def:image-preserving-hashes}.
If $\beta\geq1$, the stage above extends the proper partial coloring
and leaves maximum uncolored degree less than $\Delta'/2$ in
$O(R\sqrt M+\log^*\Delta)$ rounds. It works in \CONGEST when
$C,M\leq n^{O(1)}$.
\end{lemma}

\begin{proof}
The image guarantee implies $M>k/2=\Delta'/4$, so the field choice
gives $q=\Theta(\sqrt M)$. Since $\beta\geq1$, we have
$\beta=\Theta(\Delta'/q)$, and \cref{thm:partition} produces
$p=O(q)$ parts. The hypotheses of
\cref{lem:algebraic} hold in every call, since $q$ accommodates both
the image colors and the auxiliary colors, and each active node has
\[
    |h_j(L'_v)|>k/2=\Delta'/4\geq4q\beta.
\]
Adjacent nodes colored in the same call receive different images
under the same function and hence
choose different actual colors. These colors also avoid all previously
colored neighbors, so each call extends the proper partial coloring.

To see that the degree is halved, suppose a node $v$ remains uncolored
with a final residual list $S$ of size at least $k$.
By \eqref{eq:hash-property}, some function $h_j$ satisfies
$|h_j(S)|>k/2$. When $h_j$ was tried in $v$'s part, its
residual list still contained $S$, since lists only shrink. Thus, $v$
would have been active and colored, a contradiction. Every remaining
list has size less than $k$, so the maximum uncolored degree is less
than $k=\Delta'/2$.

The auxiliary color reduction and decomposition take
$O(q+\log^*\Delta)$ rounds. Each part uses $R$ two-round calls,
so processing all $p$ parts takes $O(pR)=O(Rq)$ rounds.
Messages contain only activity flags, auxiliary colors, field
elements, or actual colors, and have $O(\log C+\log M)$ bits.
Hash tables and lists are computed locally and are never transmitted.
\end{proof}

\paragraph{The complete algorithm.}
We use the family of \cref{lem:hashes}, with $M=8k=4\Delta'$ and
$R=O(\log(1+C/\Delta'))$. Thus $q=\Theta(\sqrt{\Delta'})$, and
$\beta\geq1$ whenever $\Delta'$ is at least a sufficiently large absolute
constant. By \cref{lem:hash-stage}, each stage then takes
\[
    O\left(\sqrt{\Delta'}\log\left(1+\frac C{\Delta'}\right)
        +\log^*\Delta\right)
\]
rounds.
Initially, every node is uncolored. Compute and retain a proper
auxiliary $O(\Delta^2)$-coloring by \cref{thm:linial}. Let
$\Delta'_0$ be the least power of two at least $\Delta$, and
set $\Delta'=\Delta'_0$. Repeat the stage above, replacing
$\Delta'$ by $\Delta'/2$ after each stage, until it falls
below this constant threshold. On the remaining graph, use
\cref{thm:linial} with degree bound $\Delta'$ and the retained
auxiliary coloring. Process the resulting constant number of color
classes in order, choosing and announcing an available actual color
at each node.

\begin{proof}[Proof of \cref{thm:main}]
Initially, the maximum degree is at most $\Delta'_0<2\Delta$.
Every executed stage has even $k=\Delta'/2\geq2$ and
$k\leq\Delta'_0/2\leq\Delta<C$. Thus,
\cref{lem:hash-stage} applies and justifies each update of the degree
bound. The final color classes are independent and their residual
lists are nonempty, so the finishing procedure colors all remaining
nodes. This also covers bounded values of $\Delta$ for which no stage
is executed.

At stage $i$, the degree bound is
$\Delta'_i=\Delta'_0/2^i$. Extending the sum to all
$i\geq0$, the degree-dependent part of the running time is at most
\begin{align*}
    \sum_{i\geq0}O\left(\sqrt{\Delta'_i}
        \log\left(1+\frac C{\Delta'_i}\right)\right)
    &\leq O(\sqrt\Delta)\sum_{i\geq0}2^{-i/2}
        \left(\log\left(1+\frac C\Delta\right)+i\right)\\
    &=O\left(\sqrt\Delta\log\left(1+\frac C\Delta\right)\right).
\end{align*}
Both $\sum_i2^{-i/2}$ and $\sum_i i2^{-i/2}$ converge.
The repeated auxiliary color reductions add at most
\[
    O\bigl(\log\Delta\,\log^*\Delta\bigr)
    =O(\sqrt\Delta)
\]
rounds, which are absorbed by the same bound. The final graph has
constant degree, so its auxiliary color reduction and greedy
completion take $O(\log^*\Delta)$ rounds. Adding the initial
$O(\log^* n)$ rounds proves the claimed complexity.

All messages have size $O(\log n)$ when $C\leq n^{O(1)}$, so the
algorithm works in \CONGEST. For arbitrary
finite $C$, it works in \LOCAL. All choices are deterministic, and
the search of \cref{lem:hashes} gives every node the same hash family
from $C$ and $k$.
\end{proof}

\section*{AI Disclosure}
\addcontentsline{toc}{section}{AI Disclosure}
The main ideas of this paper were developed by the authors. We used
ChatGPT (5.6sol, 6 Astra, and 6.1sol) to further explore these ideas.
The large language models were particularly useful in defining the hash
families. We also used them to assist with the writing process.
All formulations and proofs were carefully checked by the authors.

\clearpage
\setlength{\parskip}{0pt plus 1pt}
\bibliographystyle{plainnat}
\let\originalbibsection\bibsection
\renewcommand{\bibsection}{%
    \originalbibsection
    \addcontentsline{toc}{section}{\refname}%
}
\bibliography{references}

\clearpage
\appendix
\section{Explicit Hash Functions}
\label{sec:explicit-hashing}

The hash family of \cref{lem:hashes} has a range of size $O(k)$,
but finding it by exhaustive search can require substantial local
computation. In this appendix, we replace this search by the explicit
lossless expanders
of Capalbo, Reingold, Vadhan, and Wigderson~\cite{CRVW2002}.
Both the number of functions and the ratio between the range size
and $k$ now grow slowly with $C/k$. The resulting bound is
weaker for large color spaces, but gives the same round complexity
as \cref{thm:main} when $C=O(\Delta)$.

\begin{theorem}[Coloring with explicit hash functions]
\label{thm:explicit-coloring}
The $(\deg+1)$-list-coloring problem with colors from $[C]$, where
$C\geq\Delta+1$, can be solved deterministically in
\[
    \sqrt\Delta\,
    2^{O((\log\log(4+C/\Delta))^3)}+O(\log^* n)
\]
rounds in \LOCAL, and in \CONGEST when $C\leq n^{O(1)}$.
The algorithm uses explicit hash functions, each of which can be
evaluated in time polynomial in $\log C$ from the common parameters
and its index.
In particular, for $C=O(\Delta)$, the round complexity is
$O(\sqrt\Delta+\log^* n)$.
\end{theorem}

\subsection{From Lossless Expanders to Hash Functions}
\label{sec:explicit-family}

We use a bipartite graph with $N$ left vertices, $M$ right vertices,
and left degree $R$. Let $E(x,j)$ denote the $j$th neighbor of a
left vertex $x$. For sets $A\subseteq[N]$ and $B\subseteq[R]$, define
$E(A,B):=\{E(x,j):x\in A,\ j\in B\}$.
A lossless expander guarantees that small sets of
left vertices have many distinct neighbors. We will use the neighbor
functions $x\mapsto E(x,j)$ as hash functions. The expander is only used
to define these functions; it is unrelated to the communication graph.

\begin{theorem}[Explicit lossless expanders~\cite{CRVW2002}]
\label{thm:crvw}
There are absolute constants $a_0>0$ and $0<c\leq1$ such that, for
powers of two $2\leq M\leq N$, there is an explicit function
$E:[N]\times[R]\to[M]$ with
\[
    R\leq 2^{a_0(\log(2+\log(N/M)))^3}
\]
and
\begin{equation}
\label{eq:lossless-expansion}
    |E(S,[R])|\geq\tfrac34 R|S|
    \qquad\text{whenever }|S|\leq cM/R.
\end{equation}
Given $N,M,x,j$, the value $E(x,j)$ can be computed in time
polynomial in $\log N$.
\end{theorem}

\begin{proof}
We use the lossless-condenser/expander correspondence of Ta-Shma,
Umans, and Zuckerman~\cite[Lemma~2.2.1]{TaShmaUmansZuckerman2001},
also stated after Definition~3.4 in~\cite{CRVW2002}.
We spell out the needed direction.
For $N=M$, take $R=1$ and $E(x,1)=x$.
Otherwise, put $n=\log N$, $m=\log M$, and
$t=n-m=\log(N/M)\geq1$.
Theorem~7.3 of~\cite{CRVW2002}, with $\varepsilon=1/4$, gives
a lossless conductor
$E:\{0,1\}^n\times\{0,1\}^d\to\{0,1\}^m$ with
\[
    d=O((\log(4t))^3)=O((\log(t+2))^3),\qquad
    k_{\max}=m-d-2-O(1).
\]
Identify the bit strings with $[N]$, $[R]$, and $[M]$, where $R=2^d$.
The bound on $d$ gives the claimed degree bound for an absolute $a_0$,
and evaluation takes $\operatorname{poly}(n)=\operatorname{poly}(\log N)$ time.
Choose an absolute constant $b\geq0$ such that $k_{\max}\geq m-d-b$,
and set $c=2^{-b}$. Then $\log(cM/R)=m-d-b\leq k_{\max}$.

Let $S\subseteq[N]$ be nonempty with $|S|\leq cM/R$.
Take $X$ uniform on $S$ and $J$ independent and uniform on $[R]$.
The min-entropy of $X$ is $k=\log|S|\leq k_{\max}$.
By the lossless-conductor property (Definition~3.4 in~\cite{CRVW2002}),
the distribution of $E(X,J)$ is within statistical distance $1/4$
of a distribution $Q$ on $[M]$ with
$Q(y)\leq2^{-(k+d)}=1/(R|S|)$ for every $y\in[M]$.\footnote{Statistical
distance at most $1/4$ means that the
probabilities assigned to any event differ by at most $1/4$:
$|\Pr[E(X,J)\in A]-Q(A)|\leq1/4$ for every $A\subseteq[M]$,
where $Q(A)=\sum_{y\in A}Q(y)$.}
Since $E(X,J)$ lies in $T:=E(S,[R])$ with probability one,
\[
    \tfrac34\leq Q(T)=\sum_{y\in T}Q(y)\leq\frac{|T|}{R|S|}.
\]
This proves \eqref{eq:lossless-expansion}; the empty set satisfies it trivially.
\end{proof}

For every positive integer $a$, define, for $t\geq1$,
\begin{equation}
\label{eq:explicit-overhead}
    F_a(t):=2^{a\lceil\log\log(4+t)\rceil^3}.
\end{equation}
For the constants $a_0$ and $c$ from \cref{thm:crvw}, we can choose
an absolute positive integer $a$ such that
\begin{equation}
\label{eq:explicit-overhead-choice}
    2^{a_0(\log(3+\log t))^3}\leq cF_a(t)
    \qquad\text{for all }t\geq1.
\end{equation}
Such a choice exists because the two cubed expressions differ by
at most a constant factor and are bounded away from zero.
The ceiling makes $F_a(t)$ an integer, and for each fixed $a$,
$F_a(t)=2^{O((\log\log(4+t))^3)}$.

\begin{lemma}[Explicit image-preserving hash family]
\label{lem:explicit-hashes}
There is an absolute positive integer $a$ such that, for all integers
$N\geq2$ and $1\leq k\leq N$, there is an explicitly computable
image-preserving family at threshold $k$
(\cref{def:image-preserving-hashes}) of functions
$h_1,\ldots,h_R:[N]\to[M]$ with
\[
    R\leq F_a(N/k),\qquad
    k\leq M\leq 2kF_a(N/k),\qquad M<2N.
\]
Each value $h_j(x)$ can be computed in time
polynomial in $\log N$.
\end{lemma}

\paragraph{Construction.}
Choose $a$ satisfying \eqref{eq:explicit-overhead-choice}.
Let $N'$ be the least power of two at least $N$, and let $M$ be the
smaller of $N'$ and the least power of two at least $kF_a(N/k)$.
If $M=N'$, use the single function $h_1(x)=x$.
Otherwise, construct the function $E$ of \cref{thm:crvw} for these
$N',M$, embed $[N]$ into $[N']$, and set
\[
    h_j(x)=E(x,j),\qquad j\in[R].
\]
The construction depends only on $N$ and $k$, so every node obtains
the same family.

\begin{proof}
The choice of $M$ gives all claimed range bounds, including
$M\leq N'<2N$. If $M=N'$, the image guarantee is immediate.
Otherwise, $M\geq kF_a(N/k)\geq k$ and $N'<2N$, so
\cref{thm:crvw} and \eqref{eq:explicit-overhead-choice} give
\[
    R\leq 2^{a_0(\log(3+\log(N/k)))^3}
        \leq cF_a(N/k)\leq F_a(N/k).
\]
Consequently, $cM/R\geq k$. For every $k$-element set $S$,
\[
    \sum_{j=1}^R|h_j(S)|
        \geq\left|\bigcup_{j=1}^R h_j(S)\right|
        =|E(S,[R])|\geq\tfrac34 Rk.
\]
Some image therefore has size at least $3k/4>k/2$. A larger set
contains a $k$-element subset and inherits the guarantee.
Each function is evaluated using the neighbor algorithm of
\cref{thm:crvw}, without storing a table or searching over families.
\end{proof}

\subsection{The Coloring Algorithm}
\label{sec:explicit-algorithm}

Fix $a$ as in \cref{lem:explicit-hashes}.
We use the setup of \cref{sec:compressed} with the family of
\cref{lem:explicit-hashes}, taking $N=C$ and $k=\Delta'/2$.
Compute $\psi$, $q$, and $\beta$ as there. If $\beta\geq1$, execute
the same degree-halving stage, trying all hash functions within each
part before proceeding to the next part. If $\beta=0$, process the auxiliary
color classes of $\psi$ in order, choosing and announcing an available actual color
at each node. This colors all remaining nodes and ends the algorithm.

\begin{lemma}[Degree halving with explicit hashing]
\label{lem:explicit-stage}
The stage above extends the proper partial coloring and either leaves
maximum uncolored degree less than $\Delta'/2$ or colors all
remaining nodes, in
\[
    O\left(\sqrt{\Delta'}\,
        F_a(2C/\Delta')^{3/2}+\log^*\Delta\right)
\]
rounds. It works in \CONGEST when $C\leq n^{O(1)}$.
\end{lemma}

\begin{proof}
Write $\Phi:=F_a(2C/\Delta')$ for this proof. Since
$k=\Delta'/2$ and
$k\leq M\leq2k\Phi$, the field choice gives
\begin{equation}
\label{eq:explicit-field-size}
    q=O(\sqrt{\Delta'\Phi}).
\end{equation}
If $\beta\geq1$, \cref{lem:hash-stage} gives the claimed degree
reduction and a proper extension of the coloring in
\[
    O(R\sqrt M+\log^*\Delta)
        =O\bigl(\sqrt{\Delta'}\,\Phi^{3/2}+\log^*\Delta\bigr)
\]
rounds, since $R\leq\Phi$ and $M\leq\Delta'\Phi$.

It remains to analyze the finishing procedure. Here $\beta=0$, so
$\Delta'<16q$.
Together with \eqref{eq:explicit-field-size}, this implies
$\Phi=\Omega(\Delta')$. Equivalently, there is an absolute constant
$K>0$ such that $\Delta'\leq K\Phi$. Processing the independent
auxiliary color classes takes $O((\Delta')^2)$ rounds, and
\[
    \begin{aligned}
        (\Delta')^2=\sqrt{\Delta'}\,(\Delta')^{3/2}\leq K^{3/2}\sqrt{\Delta'}\,\Phi^{3/2}=O\bigl(\sqrt{\Delta'}\,\Phi^{3/2}\bigr).
    \end{aligned}
\]
Every node has an available actual color because its residual list
remains nonempty. Thus this procedure also extends the proper partial
coloring within the claimed bound.

Since $M<2C$, the \CONGEST guarantee of \cref{lem:hash-stage}
applies when $C\leq n^{O(1)}$. The finishing procedure sends only
auxiliary labels and actual colors, so it satisfies the same message bound.
\end{proof}

\paragraph{The complete algorithm.}
Compute and retain the initial proper auxiliary
$O(\Delta^2)$-coloring by \cref{thm:linial}. Let
$\Delta'_0$ be the least power of two at least $\Delta$, and
set $\Delta'=\Delta'_0$. Repeat the stage above,
replacing $\Delta'$ by $\Delta'/2$ after each stage,
until the finishing procedure is used. For $\Delta'=2$, we
have $\beta=0$, so the algorithm always terminates.

\begin{proof}[Proof of \cref{thm:explicit-coloring}]
Initially, the maximum degree is at most
$\Delta'_0<2\Delta$. In every stage,
$1\leq k=\Delta'/2\leq\Delta'_0/2\leq\Delta<C$,
so \cref{lem:explicit-stage} applies. Each stage extends the proper
partial coloring and either halves the degree bound or colors all
remaining nodes. Since the algorithm terminates, this proves correctness.

At stage $i$, the degree bound is
$\Delta'_i=\Delta'_0/2^i$. Thus,
$2C/\Delta'_i\leq2^{i+1}C/\Delta$ and
$\Delta'_i<2\Delta\,2^{-i}$.
By \cref{lem:explicit-stage}, the degree-dependent cost is at most
\begin{equation}
\label{eq:explicit-series}
    O(\sqrt\Delta)\sum_{i\geq0}
        2^{-i/2}F_a(2^{i+1}C/\Delta)^{3/2}.
\end{equation}
This includes the finishing procedure. 
Set $t=C/\Delta\geq1$ and $\ell=\log(4+t)$. For every $i\geq0$,
\[
    \log\log(4+2^{i+1}t)
        \leq\log(\ell+i+1)
        \leq\log\ell+\log(i+2).
\]
Using $(x+y)^3\leq4(x^3+y^3)$ for nonnegative $x,y$, and
absorbing the ceiling in \eqref{eq:explicit-overhead}, we obtain
\[
    F_a(2^{i+1}t)^{3/2}
        \leq 2^{O((\log\log(4+t))^3)}
             \cdot2^{O((\log(i+2))^3)}.
\]
The constants in both exponents are absolute. Because
$(\log(i+2))^3=o(i)$, the series
\[
    \sum_{i\geq0}2^{-i/2}\,2^{O((\log(i+2))^3)}
\]
converges to an absolute constant. Substituting into
\eqref{eq:explicit-series} proves the claimed degree-dependent
bound. The auxiliary color reductions contribute
$O(\log\Delta\,\log^*\Delta)=O(\sqrt\Delta)$ rounds.
Adding the initial $O(\log^* n)$ rounds proves the claimed complexity.
The hash evaluation and message-size bounds follow from
\cref{lem:explicit-hashes,lem:explicit-stage}. All choices are
deterministic, and every node obtains the same hash family from the
common parameters $C$ and $k$.
\end{proof}

\end{document}